\documentclass[runningheads]{llncs}
\usepackage{graphicx}
\usepackage[utf8]{inputenc}
\usepackage{amsmath}
\usepackage{amsfonts}
\usepackage{amssymb}
\usepackage{authblk}
\usepackage{verbatim}
\usepackage{thmtools}
\usepackage{thm-restate}
\usepackage{todonotes}
\usepackage%[hidelinks]
			{hyperref}

\usepackage[capitalize, noabbrev]{cleveref}
\usepackage[autostyle=true]{csquotes}
\usepackage{graphicx}
\usepackage{color}
\usepackage{nicefrac}
\usepackage{enumerate}
\usepackage[ruled, lined, linesnumbered, commentsnumbered, noend]{algorithm2e}

\usepackage{tikz}
\usetikzlibrary{calc}
\usetikzlibrary{arrows,shapes}
\usetikzlibrary{decorations.markings}
\usetikzlibrary{arrows.meta}

\newcommand{\R}{\mathbb{R}}
\newcommand{\rz}{\mathbb{R}}
\newcommand{\N}{\mathbb{N}}

\newcommand{\nz}{\mathbb{N}}
\newcommand{\set}[1]{\{ #1 \}}

\newcommand{\bigO}{\mathcal{O}}

\newcommand{\const}{\textsf{source}}

\newcommand{\Pst}{\mathcal{P}_{st}}
\newcommand{\Nscript}{N}

\newcommand{\boxxx}[1]
 {\fbox{\begin{minipage}{11.80cm}\begin{center}\bigskip\begin{minipage}{11.30cm}
  #1\end{minipage}\end{center}~\end{minipage}}}

\DeclareMathOperator{\red}{\textsf{reduced}}
\DeclareMathOperator{\val}{val}

\begin{document}

\title{A linear time algorithm for linearizing quadratic and higher-order shortest path problems}

\titlerunning{A inear time algorithm for the linearization of the QSPP and SPP$_d$}

\author{Eranda \c{C}ela\inst{1}\orcidID{0000-0002-5099-8804} \and
Bettina Klinz\inst{1}\orcidID{0000-0002-6156-688X} \and
Stefan Lendl\inst{2}\orcidID{0000-0002-5660-5397} \and
{\textdagger Gerhard J. Woeginger}\footnote[0]{\textdagger Deceased in April 2022.}\inst{3}\orcidID{0000-0001-8816-2693} \and
Lasse Wulf\inst{1}\orcidID{0000-0001-7139-4092}}
\authorrunning{E. \c{C}ela, B. Klinz, S. Lendl, G. J. Woeginger, L. Wulf}
\institute{Institute of Discrete Mathematics, Graz University of Technology, Austria \email{\{cela,klinz,wulf\}@math.tugraz.at} \and
Institut of Operations and Information Systems, University of Graz, Austria
\email{stefan.lendl@uni-graz.at} \and
Department of Computer Science, RWTH Aachen, Germany}

\maketitle              % typeset the header of the contribution

\begin{abstract}
An instance of the NP-hard Quadratic Shortest Path Problem (QSPP) is called linearizable iff it is equivalent to an instance of the classic
Shortest Path Problem (SPP) on the same input digraph.
The linearization problem for the QSPP (LinQSPP) decides whether a given QSPP
instance is linearizable and determines the corresponding SPP instance in
the positive case. We provide a novel linear time algorithm for the LinQSPP on acyclic digraphs which runs considerably faster than the
previously best algorithm. The algorithm is based on a new insight revealing that the linearizability
of the QSPP for acyclic digraphs can be seen as a local property. Our
approach extends to the more general higher-order shortest path problem. 

\end{abstract}
\section{Introduction}\label{intro:sec}

In this paper we consider the linearization problem for 
nonlinear
generalizations of the  \emph{Shortest Path Problem (SPP)}, a classic
combinatorial optimization problem.
An instance of the SPP  consists  of a digraph $G = (V, A)$,  a source vertex $s \in V$, a sink vertex $t \in V$, and a
cost function $c\colon A\to \rz$, which maps each arc $a\in A$ to its cost
$c(a)$.  The cost of a simple directed $s$-$t$-path $P$, 
is given by\footnote{We use the same  notation for the path $P$ and the set of its arcs.}
\begin{equation}\label{SPP:obj}  \text{SPP}(P,c):=\sum_{a\in P} c(a)\, .  \end{equation}
The goal  is to find a simple directed 
$s$-$t$-path in $G$ which minimizes the objective 
(\ref{SPP:obj}).
In general it is assumed that 
there are no circuits of negative weight in $G$.

Consider now a number $d\in \nz$.  The \emph{Order-d Shortest Path Problem
  (SPP$_d$)}
 takes as input a  digraph $G = (V, A)$,  
  a source vertex $s \in V$, a sink vertex $t \in V$, and an order-$d$ arc interaction 
cost  function $q_d\colon \set{B \subseteq A : |B| \leq d}
\to \rz$. Thus $q_d$ assigns a  weight to every subset of
arcs of cardinality at most $d$.
 The cost of a simple directed $s$-$t$-path $P$ 
is given by
\begin{equation}\label{dSPP:obj}
  \text{SPP}_d(P,q_d):=\sum_{S\subseteq P \colon |S|\le d} q_d(S)\, .\end{equation}
The goal is to find  
a simple directed 
$s$-$t$-path in $G$ which minimizes the objective function (\ref{dSPP:obj}).
For $d=2$ we
obtain the   \emph{Quadratic  Shortest Path Problem (QSPP)} which has already been studied in the literature
~\cite{cela2021linearizable,huSo2018,huSo2020,rostami2018}. 
For notational convenience we write $\text{QSPP}(P,q)$ for $\text{SPP}_d(P,q_d)$ if $d=2$.
\smallskip

The QSPP arises 
in network optimization  problems where costs are associated with both single arcs and pairs of arcs.
This includes 
%for example
variants of stochastic and time-dependent route planing  problems  
\cite{nie2009reliable,sen2001mean,sivakumar1994variance}
and network design problems 
\cite{murakami1997restoration,gamvros2006satellite}. 
For an overview on applications of the QSPP see \cite{huSo2020,rostami2018}.
We are not  aware of any
publications  
for the case $d>2$.

While the SPP can be solved in polynomial time, the QSPP is an NP-hard
problem even for the special case of the adjacent QSPP where  the  costs of all pairs of non-consecutive  arcs are  
zero ~\cite{rostami2018}.
The QSPP is  an  extremely difficult
problem also from the practical point of view.  
 Hu and Sotirov~\cite{huSo2020} report that a state-of-the-art quadratic solver can
solve QSPP instances with up to $365$ arcs, while their  tailor-made B\&B
%branch and
%bound 
algorithm can solve instances with up to $1300$ arcs to optimality within one  hour. 
Instances of the SPP can however be solved in a fraction of a second for graphs with
millions of vertices and arcs.
\smallskip

Given the 
hardness of the QSPP, a research line on this problem has focussed
on 
polynomially solvable special cases which
arise if the input graph and/or the cost coefficients have certain specific
properties. Rostami et al.~\cite{rostami2015} have presented a polynomial time
algorithm for the adjacent QSPP in acyclic digraphs and in series-parallel graphs. Hu 
and Sotirov~\cite{huSo2018}
have shown that the QSPP can be solved in polynomial time if the  quadratic
costs build a  nonnegative symmetric product matrix, or if the quadratic costs
build a sum matrix and all $s$-$t$-paths in 
$G$ have the same number of arcs. 
\smallskip

These two polynomially solvable
special cases of the QSPP belong to the larger class of the \emph{linearizable $\text{SPP}_d$ instances} defined as follows.
\begin{definition}
 An instance of the SPP$_d$ with an input digraph $G=(V,A)$, a source node $s$, a sink node $t$
 and a  cost function $q_d$ is called linearizable if there exists a cost function
 $c\colon A\to \rz_+$ such that for any simple directed $s$-$t$-path  $P$ in $G$ the equality
 $\text{SPP}(P,c) = \text{SPP}_d(P,q_d)$ holds.
  A linearizable instance  of the  QSPP is defined analogously, just  replacing $\text{SPP}_d(P,q_d)$ by $\text{QSPP}(P,q)$.
\end{definition}
The recognition of linearizable QSPP (SPP$_d$) instances, also
called \emph{the linearization problem for the QSPP (SPP$_d$)}, abbreviated by \textsc{Lin}QSPP (\textsc{Lin}SPP$_d$) arises 
as a natural question. In this problem the task 
consists of deciding whether a
given  instance of the  QSPP (SPP$_d$)  is linearizable and in finding the linear cost function
$c$ in the positive case.
The 
notion of 
linearizable special cases of hard
combinatorial optimization problems  goes back to
 Bookhold~\cite{bookhold1990contribution} who introduced 
 it for  the quadratic assignment problem (QAP). 
 For symmetric linearizable QAP instances a full characterization has been obtained while
 only partial results are available for the linearizability of the general QAP, see
 \cite{CeDeWo2016,Erdogan2006,ErTa2007,ErTa2011,kabadi2011n,punnen2013linear}.
The linearization problem has been studied for several other quadratic combinatorial optimization problems, see  \cite{CuPu2018,sotirov2021quadratic}
for the quadratic minimum spanning tree problem,  \cite{PuWaWo2017} for the quadratic TSP, \cite{deMeSo2020} for the quadratic cycle cover problem and  \cite{huSo2021} for general binary quadratic programs.  
Linearizable instances of  a quadratic problem can be used  to generate lower bounds  needed   in B\&B algorithms. For example, Hu and Sotirov introduce the family of the so-called \emph{linearization-based bounds} \cite{huSo2021} for the binary quadratic problem. Each specific bound of this family is based on   a set of linearizable instances of the problem. The authors show that 
well-known bounds from the literature are special cases of the newly introduced bounds.
Clearly, fast algorithm for the  linearization problem are important in 
%the context of these applications.
this context.
\smallskip

While \textsc{Lin}SPP$_d$ has not been investigated in the literature so far (to the best of our knowledge), the  \textsc{Lin}QSPP has been subject of investigation in some recent papers.
In \cite{cela2021linearizable} \c{C}ela, Klinz, Lendl, Orlin, Woeginger and
Wulf proved  that it  is  coNP-complete to decide whether a   QSPP instance on an   input graph
containing   a directed cycle is linearizable.    Thus, a nice characterization of linearizable QSPP
  instances for such  graphs
  %with an input graph containing  a directed cycle 
  seems to be
  unlikely.
   In the acyclic case, Hu and Sotirov first described a polynomial-time
   algorithm for the \textsc{Lin}QSPP  on  directed
   two-dimensional grid graphs \cite{huSo2018}. Recently, in  \cite{huSo2021} they  generalized this
   result to all acyclic digraphs and proposed  an algorithm which solves the
   problem in $\bigO(nm^3)$, where $n$ and $m$ denote the number of vertices and
   arcs in 
   %the input graph
   $G$.
   %, respectively. 
\smallskip
   
   Finally, let us mention a related  concept, the  so-called universal linearizability,
   studied in \cite{cela2021linearizable,huSo2018}.
A digraph $G$ is called \emph{universally linearizable with respect to the
  QSPP} iff every instance of the QSPP on the input graph $G$ is linearizable for
every choice of the cost function $q$. In \cite{huSo2018} it is shown that a
particular class of grid graphs is universally linearizable. In
\cite{cela2021linearizable} a characterization of universally linearizable grid
graphs in terms of structural properties of the set of $s$-$t$-paths is
given. Moreoever, for acyclic digraphs a forbidden subgraphs characterization of
the universal linearizability is given in \cite{cela2021linearizable}.  
\medskip

\textbf{Contribution and organization of the paper.}
In this paper we provide a novel and simple characterization of linearizable QSPP instances on acyclic %input
digraphs.
Our characterization shows that the linearizability can be seen as a  \emph{local} property.  In particular, we show  that an instance of the QSPP on an acyclic 
%input 
digraph $G$  is linearizable if and only if each subinstance obtained by considering a subdigraph of $G$ consisting of two $s$-$t$-paths in $G$   is linearizable. Our simple characterization also works for the SPP$_d$  and even for completely arbitrary cost  functions, which  assign some cost $f(P)$ to every $s$-$t$-path $P$ without any further restrictions. The latter problem is referred to as the \emph{Generic Shortest Path Problem} (GSPP) and is formally introduced in Section~\ref{defi:sec}. Indeed, the characterization of the linearizable instances of the SPP$_d$ 
follows from the characterization of the linearizable instances of the GSPP, both on acyclic digraphs.

Further,  we propose  a %polynomial  time
linear time algorithm
which can check the local condition mentioned above   for the QSPP and the SPP$_d$. 
We note that this is not straightforward, because the number of the subinstances for which the condition needs to be checked is in general exponential. As a side result our approach reveals  an interesting connection between the \textsc{Lin}QSPP and the problem of deciding  whether all $s$-$t$-paths in a digraph have the same length.
As a result, we obtain an algorithm which solves the \textsc{Lin}QSPP linearization  in $\bigO(m^2)$ time, thus improving the best previously known running time of $\bigO(nm^3)$ obtained in \cite{huSo2021}. Our approach yields an $\bigO(d^2 m^d)$ time algorithm for the \textsc{Lin}SPP$_d$, thus providing  the first (polynomial time) algorithm for this problem.  Note that the running time of the proposed algorithms is linear in the input size for both problems,  \textsc{Lin}QSPP and \textsc{Lin}SPP$_d$, respectively. 
%Indeed, the costs of all $\Omega(m^2)$ pairs of arcs, in the case of the QSPP, and the  costs of all $\Omega(m^d)$ subsets of  arcs of cardinality %$d$, in the case of the
%SPP$_d$,   need to be encoded in the input. 

Finally, we also obtain a polynomial time algorithm that given an acyclic digraph $G$ computes a basis of the subspace of all linearizable degree-$d$ cost functions on $G$. Such a basis can be used to obtain better linearization-based bounds usable in B\&B algorithms.
\smallskip

The paper is organized as follows. 
After introducing some notations and preliminaries in Section~\ref{defi:sec} we present the result on the characterization of the linearizable QSPP and SPP$_d$ instances on acyclic input digraphs in Section~\ref{charact:sec}. The algorithms for the linearization problems \textsc{Lin}QSPP and \textsc{Lin}SPP$_d$ are presented in Section~\ref{algo:sec}. 
\cref{sec:subspace} deals with 
computing a basis of the subspace of all linearizable $d$-degree cost functions on an acyclic digraph $G$.

\section{Notations  and preliminaries}\label{defi:sec}
 Given a digraph $G=(V,A)$, a simple directed $s$-$t$-path $P$ in $G$ is specified as a sequence of
arcs $P=(a_1,a_2,\ldots,a_p)$ such that  $a_1$ starts at $s$, $a_p$
ends at $t$, nonconsecutive arcs do not share a vertex and  the end vertex of $a_{i}$ coincides with  the start vertex of $a_{i+1}$ for any $i\in \{1,\ldots,p-1\}$. The number
$p$ of arcs in $P$ is called the length of the path. We sometimes use the same notation for a path $P$ and the  set of its arcs. 
%Alternatively, a simple directed $s$-$t$-path $P$ of length $p$ in $G$  is specified as a sequence of pairwise disjoint vertices $(u_0,u_1,\dots,u_p)$, 
%where $x_0=s$, $u_p=t$ and $(u_i,u_{i+1})\in A$ for all $i\in\{0,\dots,p\}$. Thus  $a_i=(u_{i-1},u_i)$, for $i\in \{1,\dots,p\}$. 
We  consider a single arc $(x, y)$ as an $x$-$y$-path of length $1$  and a single vertex $x$ as a  trivial $x$-$x$-path of length $0$. Given an  $x$-$y$-path $P_1$ and a  $y$-$z$-path $P_2$, we denote the \emph{concatenation} of $P_1$ and $P_2$ by $P_1 \cdot P_2$. We also consider concatenations of paths and arcs, that is, terms of the form $P \cdot a$ for some $x$-$y$-path $P$ and some arc $a = (y, z)$.
 
In the linearization problem, we are concerned with acyclic digraphs $G =(V, A)$ with a source vertex $s$ and a sink vertex $t$. 
We denote by  $\Pst$ the set of all simple directed $s$-$t$-paths.
We often assume that $G$ is \emph{$\Pst$-covered}, that is, every arc in $G$ is traversed by at least one path in $\Pst$. 
It is easy to see that this assumption can be made without loss of generality. 

Let $d \geq 2$ be a natural number. The \emph{Order-$d$ interaction costs} are given by a mapping $q_d\colon \set{B \subseteq A : |B| \leq d}\to \rz$, assigning a (potentially negative) interaction cost to every subset of at most $d$ arcs.
The cost $\text{SPP}_d(P,q_d)$ of some path $P$ under interaction costs $q_d$ is defined as in equation (\ref{dSPP:obj}).
If $d$ is unambiguously clear form the context, we use the  more compact notation $f_q(P): = \text{SPP}_d(P,q_d)$.
 In this paper we explicitly allow the case $q(\emptyset) \neq 0$, because this simplifies the calculations.  
The \emph{linearization problem for the Order-$d$ Shortest Path Problem} (\textsc{Lin}SPP$_d$)
is formally defined as follows.
%  in \cref{fig:problem}.
%%%%%%%%%%%%%%
% \begin{figure}[htb]
\begin{center}
%%%%%%%%%%%%%%%%%
\boxxx{\textbf{Problem:} The  linearization problem for the SPP$_d$ (\textsc{Lin}SPP$_d$)
\\[1.0ex]
\textbf{Instance:} A $\Pst$-covered directed graph $G=(V,A)$ with $s,t\in V$, $s\neq t$; an integer $d \geq 2$;
an order-$d$ arc interaction cost function $q_d: \set{B \subseteq A \colon |B| \leq d} \to \R$.
\\[1.0ex]
\textbf{Question:} Find a \emph{linearizing cost function}  $c\colon A\to\R$  such that $\text{SPP}_d(P,q_d) = \text{SPP}(P,c)$ for all $P \in \Pst$ or decide that such a linearizing  cost function does not exist.}
%%%%%%%%%%%%%%%%%
\end{center}
% \caption{The linearization problem for the Order-$d$ Shortest Path Problem.}
% \label{fig:problem}
% \end{figure}
%%%%%%%%%%%%
In the special case $d = 2$, we obtain the    linearization problem for the QSPP (\textsc{Lin}QSPP).
\smallskip

Finally, let us consider the \emph{Generic Shortest Path Problem} (GSPP) which takes as input  a digraph $G=(V,A)$ with a source vertex $s$, a sink vertex $t$, $s\neq t$, and a generic cost function $f\colon \Pst \to \rz$  assigning  a cost $f(P)$ to every path $P\in \Pst$\footnote{We assume that $f$ is specified by an oracle.} We assume w.l.o.g.\ that $G$ is $\Pst$-covered. The goal is to find an $s$-$t$-path  which minimizes the objective function $f(P)$ over $\Pst$. 
A linearizable instance of  the GSPP and the linearization problem for the GSPP (\textsc{Lin}GSPP) are  defined analogously as in the respective definitions for SPP$_d$.
\section{A characterization of linearizable instances of  the GSPP on acyclic digraphs}\label{charact:sec}
 %%%%%%%%%%%%%%%
 The main result of this section is \cref{thm_linearization_characterization}, our novel characterization of  linearizable instances of the GSPP on acyclic digraphs defined as in Section~\ref{defi:sec}. 
 %%%%%%%%%%%%
\begin{definition}
Let $G=(V,A)$ be a $\Pst$-covered acyclic digraph. For some vertex $v$, let $P_1, P_2$ be two $s$-$v$-paths, and let $Q_1, Q_2$ be two $v$-$t$-paths. The 5-tuple $(v,P_1,P_2,Q_1,Q_2)$ is called a \emph{two-path system} contained in $G$. The system is called \emph{linearizable} with respect to the function $f : \Pst \rightarrow \R$, if there exists a cost function $c : A \rightarrow \R$ such that for all four paths $P \in \set{P_1 \cdot Q_1, P_1 \cdot Q_2, P_2 \cdot Q_1, P_2 \cdot Q_2}$ we have $f(P) = SPP(P,c)$. Such a $c$ is called  a {\emph linearizing cost function} for  $(v,P_1,P_2,Q_1,Q_2)$ with respect to $f$. 
\end{definition}

%--------------------------------------------------------------------
%Create custom tikz style for directed edges with arrow in the middle
\tikzset{->-/.style={
	decoration={
 		markings,
  		mark=at position #1 with {\arrow[scale=2,>=stealth]{>}}
  		},
  	postaction={decorate}
  },
  ->-/.default=0.5
}%--------------------------------------------------------------------
\tikzstyle{vertex}=[draw,circle,fill=black, minimum size=4pt,inner sep=0pt]
\tikzstyle{edge} = [draw,thick,-]
\tikzstyle{weight} = [font=\small]
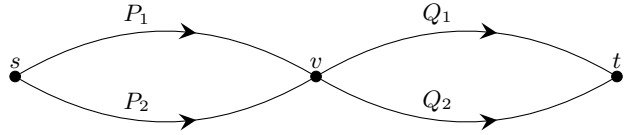
\begin{figure}[bth]
\centering
\begin{tikzpicture}[scale=1.0, auto,swap]

    \node[vertex] (s) at (0,0) {};
    \node[vertex] (v) at (4,0) {};
    \node[vertex] (t) at (8,0) {};

    \node[above] at (s) {$s$};
    \node[above] at (t) {$t$};
    \node[above] at (v) {$v$};
    
    \draw[->-=0.6,pos=0.4] (s) to[bend left] node[above]{$P_1$} (v);
    \draw[->-=0.6,pos=0.4] (s) to[bend right] node[above]{$P_2$} (v);
    \draw[->-=0.6,pos=0.4] (v) to[bend left] node[above]{$Q_1$} (t);
    \draw[->-=0.6,pos=0.4] (v) to[bend right] node[above]{$Q_2$} (t);
\end{tikzpicture}
\caption{A two-path system.}
 \label{fig:two-path-system}
\end{figure}

See \cref{fig:two-path-system} for an illustration of a two-path system. 
Note that $P_1$ and $P_2$ (as well as $Q_1$ and $Q_2$) can have common inner vertices and that the cases $P_1=P_2$, $Q_1=Q_2$, $v = s$ and $v = t$ are allowed. 
However, due to the acyclicity of $G$, the paths $P_i$ and $Q_j$ have only the vertex $v$ in common for $i,j \in\{1,2\}$. Further, observe that the linearizability of a two-path system 
is a local property, in the sense  that it only   depends on the four paths $P_1 \cdot Q_1, P_1 \cdot Q_2, P_2 \cdot Q_1$ and $P_2 \cdot Q_2$. Indeed,  the following simple characterization holds. 

\begin{proposition}
\label{obs:linearizability-two-paths}
A two-path system $(v,P_1,P_2,Q_1,Q_2)$ is linearizable with respect to some function $f\colon \Pst\to \rz$ iff
\begin{equation}
f(P_1 \cdot Q_1) + f(P_2 \cdot Q_2) = f(P_1 \cdot Q_2) + f(P_2 \cdot Q_1).    \label{eq:two-path-lin}
\end{equation}
\end{proposition}
\begin{proof}
    First, assume that $(v,P_1,P_2,Q_1,Q_2)$ is linearizable and let $c$ be the corresponding  linearizing cost function.  Let $M_1$ ($M_2$) be the multiset  resulting from   the union of the sets of the arcs of the paths $P_1\cdot Q_1$ and $P_2\cdot Q_2$  ($P_1\cdot Q_2$ and $P_2 \cdot Q_1$). Since $M_1$ and $M_2$ coincide we get   $c(P_1 \cdot Q_1) + c(P_2 \cdot Q_2) =\sum_{a\in M_1} c(a)= \sum_{a\in M_2} c(a)=c(P_1 \cdot Q_2) + c(P_2 \cdot Q_1)$. Then,   (\ref{eq:two-path-lin}) follows from the definition of the linearizability of   $(v,P_1,P_2,Q_1,Q_2)$.
    
    Assume now  that \cref{eq:two-path-lin} is true. We  show the linearizability of  the two-path system with respect to $f$ by constructing a linearizing cost function $c$.   It is easy to find a suitable  $c$ if $P_1=P_2$ or $Q_1=Q_2$. So let us consider  the more general case  where $P_1\neq P_2$ and $Q_1\neq Q_2$.  In this case, for   each  $P\in \{P_1,P_2,Q_1,Q_2\}$ there exists a so-called \emph{representative arc} $a\in P$ such that $a$  is not contained in any other path  $Q\in \{P_1,P_2,Q_1,Q_2\}$, $Q \neq P$. Let $a_1$, $a_2$, $e_1$, $e_2$ be the representative arcs of $P_1$, $P_2$, $Q_1$ and $Q_2$, respectively.  Consider now a cost function $c : A \rightarrow \R$, such that $c(a) = 0$ if  $a\not\in \{a_1, a_2, e_1, e_2\}$, and $c(a_1)$, $c(a_2)$, $c(e_1)$, 
$c(e_2)$ fulfill the following linear equations:
    \begin{equation*}
        \begin{array}{llcllcl}
        c(a_1) & & + & c(e_1) & & = f(P_1Q_1) \\
        c(a_1) & & + & & c(e_2) &= f(P_1Q_2) \\
        &c(a_2) & + & c(e_1) & &= f(P_2Q_1) \\
        &c(a_2) & + & & c(e_2) &= f(P_2Q_2) 
    \end{array}
    \end{equation*}
    Using basic linear algebra, one can see that this system indeed has a solution whenever  \cref{eq:two-path-lin} holds (there is even a solution with $c(e_2) = 0$). Thus, $c$ constructed as above is a linearizing  cost function for $(v,P_1,P_2,Q_1,Q_2)$ with respect to $f$. 
    \qed
\end{proof}

Now, consider an instance of the GSPP with a $\Pst$-covered acyclic digraph $G$, with a  source vertex $s$, a sink vertex $t$ and a generic cost function $f\colon \Pst\to \rz$. When is this instance $(G,s,t,f)$ linearizable? Obviously, if $G$ contains a two-path system which is not linearizable with respect to $f$, then $(G,s,t,f)$ is not linearizable. The following theorem shows that the linearizability of each two-paths system with respect to $f$ is a sufficient condition for $(G,s,t,f)$ being linearizable. 

\begin{theorem}
\label{thm_linearization_characterization}
Let $G$ be a $\mathcal{P}_{st}$-covered acyclic digraph with a source vertex $s$ and a sink vertex $t$ and  let $f : \Pst \rightarrow \R$ be a generic cost function. Then the instance $(G,s,t,f)$ of the GSPP  is linearizable if and only if every two-path system contained in $G$ is linearizable with respect to $f$.
\end{theorem}

 Before proving the theorem, we need some preparation. Let $G = (V,A)$ be a $\Pst$-covered acyclic digraph with source vertex $s$ and sink vertex $t$. First we introduce a \emph{topological arc order} as a total 
  order $\preceq$ on $A$ such that for any pair of  arcs $a$, $a'$ in $A$ the following holds:   if there  exists a path $P$  containing both $a$ and $a'$ such that $a$ comes before $a'$ in $P$, then $a\preceq a'$.  It is easy to see  that any acyclic digraph has a (in general non-unique)  topological arc order. Moreover, a topological arc order  can be obtained from a topological vertex order. 
 
 Further, we recall the definition of a \emph{system of nonbasic arcs}  introduced by  Sotirov and Hu~\cite{huSo2021}.  
 \begin{definition}\label{nonbasic:def}
 Let $G$ be a $\mathcal{P}_{st}$-covered acyclic digraph with a source vertex $s$ and a sink vertex $t$. A set $\Nscript \subseteq A$ is called a \emph{system of nonbasic arcs}, iff for every vertex $v \in V \setminus \set{s,t}$ exactly one of the arcs starting at $v$ is contained in $\Nscript$. The latter  arc is called the \emph{nonbasic arc of $v$}. An arc $a \in A \setminus N$ is called \emph{basic}.
 \end{definition}
 Obviously, the system of nonbasic arcs is not unique.  Any such system  forms an in-tree rooted at $t$ containing   all the vertices in $V$ except for $s$. For some system of nonbasic arcs $\Nscript$ and some vertex $v \in V \setminus \set{s}$, we let $N_v$ denote the unique $v$-$t$-path consisting  of nonbasic arcs (where $N_t$ is the trivial path).  A  cost function $c\colon A \rightarrow \R$ is called \emph{in reduced form} with respect to $\Nscript$, if $c(a) = 0$ for all nonbasic arcs $a \in \Nscript$. The following lemma is an easy adaption from \cite{huSo2021}, where an analogous statement was proven for the less general case of the QSPP instead of the GSPP (details are provided in the full version of this paper).

\begin{lemma}[adapted from {\cite[Prop. 4]{huSo2021}}]
\label{lemma:nonbasic-arcs}
Let $G$ be a $\Pst$-covered acyclic digraph with a source vertex $s$ and a sink vertex $t$. Let  $f\colon  \Pst \rightarrow \R$ be a generic cost function and  let $\Nscript \subseteq A$ be a fixed system of nonbasic arcs. If  $(G,s,t,f)$ is a linearizable instance of the GSPP, then there exists one and only one linear cost function $c\colon A \rightarrow \R$ which is both a linearizing cost function  and in reduced form.
\end{lemma}
% \begin{proof}
%     We have to prove both existence and uniqueness. For the existence, assume there exists a linearization $c : A \rightarrow \R$. Consider some vertex $v \in V \setminus \set{s,t}$ and its nonbasic arc $a_v$. Consider the following modification of the function $c$: Let $\beta = c(a_v)$, then reduce the cost of each outgoing arc of $v$ by $\beta$, and increase the cost of each incoming arc of $v$ by $\beta$. This operation sets the cost of $a_v$ to $0$ and does not change the linear cost of any $s$-$t$-path. Now let $v_1,\dots,v_n$ be a topological vertex order with $v_1 = s$ and $v_n = t$. We repeat the described operation for every vertex $v_{n-1}, v_{n-2}, \dots, v_2$  in this order. It is easily verified that the obtained cost function is a linearization of $(G, f)$ and is in reduced form.
    
%     For the uniqueness, assume that there are two distinct linearization vectors $c, c' : A \rightarrow \R$ with the property that all nonbasic arcs have value 0. Consider some topological arc order $\preceq$ and let $a = (u,v)$ be the first arc in the order such that $c(a) \neq c'(a)$. There exists an $s$-$u$-path $P$. The path $R := P \cdot a \cdot N_v$ is an $s$-$t$-path. By assumption, we have $c(P) = c'(P)$ and $c(N_v) = c'(N_v) = 0$. But then $c(R) \neq c'(R)$, a contradiction. 
%     \qed
% \end{proof}

Let $(G,s,t,f)$ be a linearizable  instance of the GSPP with $G=(V,A)$ and $\Nscript\subseteq A$ be a fixed system of nonbasic arcs. For a  linearizing cost function $c \colon A \rightarrow \R$, we denote by $\red(c)$ the unique linearizing cost function  in reduced form (which exists due to \cref{lemma:nonbasic-arcs}). 
It follows from the arguments in the proof of  \cref{lemma:nonbasic-arcs} that for given $c$ one can compute $\red(c)$ in $\bigO(n +m)$ time. We are now ready to sketch the proof of our main theorem.
\begin{proof}[Sketch of the proof of \cref{thm_linearization_characterization}]

The necessity of the conditions for linearizability is trivial.
Now we prove the sufficiency.
Thus we  assume that every two-path system is linearizable with respect to $f$ and show that $(G,s,t,f)$ is linearizable. Let $\Nscript$ be a system of nonbasic arcs. The main idea is to find a linearizing cost function  which is in reduced form, i.e., which has value 0 on all nonbasic arcs. 
    To this end we consider a topological arc order $\preceq$ on the set  $A$ of arcs in $G$ and inductively define a linearizing cost function $c \colon A \rightarrow \R$ as follows. 
    For any arc $a = (u,v)$  set
        \begin{equation}\label{lincost:equ} c(a) := \begin{cases}
        f(P \cdot a \cdot N_v) - \sum_{a' \in P}c(a')%; 
        & a \not\in \Nscript \\
        0%
        ; & a \in \Nscript
        \end{cases}
        \end{equation}
   for some $s$-$u$-path $P$. 
   \smallskip
   
   Consider now  the following claim the proof of which is omitted for brevity. 
   \textbf{Claim:} If all two-path systems in $G$ are linearizable with respect to $f$, then function $c$ in Equation~(\ref{lincost:equ}) is well-defined and independent of the concrete choice of $P$. Moreoever, the following equation holds for all $s$-$u$-paths $P$:
    \begin{equation}\label{claim2:equ} f(P \cdot a \cdot N_v)=c(a)+\sum_{a' \in P} c(a') = c(P \cdot a \cdot N_v)\end{equation}
    
     Observe that  the claim immediately implies  that $(G,s,t, f)$ is linearizable. Indeed, let $c$ be the  cost function defined in \cref{lincost:equ} and  let $Q$ be some $s$-$t$-path. Choose $a = (x,t)$ to be the  last arc on $Q$. Then $N_t$ is the trivial path from $t$ to $t$, so by  applying \cref{claim2:equ} to the arc $a$, we have $f(Q) = c(Q)$. 

    \qed
\end{proof}

%This concludes the proof of our main theorem. 
%How can this theorem be used to construct an algorithm deciding linearizability?
Since in general a graph contains exponentially many different two-path systems,   \cref{thm_linearization_characterization} does not seem to lead to  an  efficient algorithm for  the linearization problem \textsc{Lin}GSPP at a first glance. However, we show in the next section that this is indeed the case. The arguments  are based on a more technical version of \cref{thm_linearization_characterization} and involve the concept of  so-called \emph{strongly basic arcs} and their property $(\pi)$ defined below.

\begin{definition}\label{stronglybasic:def}
    Let $G = (V, A)$ be an acyclic $\Pst$-covered digraph with source vertex $s$ and sink vertex $t$.  Let $f \colon  \Pst \rightarrow \R$ be a generic cost function and let $N \subseteq A$ be a system of nonbasic arcs in $G$. A basic arc $(u, v)$ is called \emph{strongly basic}, if it is not incident to the source vertex, that is if $u \neq s$.
    
  \noindent  A strongly basic arc $a = (u,v)$ has the \emph{property $(\pi)$}, if for 
    any $s$-$u$-paths $P$ the value $ \val(a, P) :=  f(P \cdot a \cdot N_v) - f(P \cdot N_u)$ does not depend  on the choice of $P$. 
    \end{definition}
    
    Thus, if a strongly basic arc $a = (u,v)$ has the property $(\pi)$,   we have  $\val(a, P) = \val(a, Q)$ for any two  $s$-$u$-paths $P, Q$ and this implies the existence of a value   $\val(a) := \val(a, P)$ for each  $s$-$u$-path $P$ and $val(a)$ is well defined for each strongly basic arc. Finally, we set   $\val(a) := f(a \cdot N_v)$ for each  basic arc $a = (s, v)$.

\begin{lemma}
\label{lemma:property-pi}
    Let $G = (V, A)$ be an acyclic $\Pst$-covered digraph with source vertex $s$ and sink vertex $t$. Let $f \colon \Pst \rightarrow \R$ be a generic cost function and let $N \subseteq A$ be a system of nonbasic arcs in $G$. Then $(G, s,t,f)$ is linearizable if and only if every strongly basic arc has the property $(\pi)$. In this case, the mapping   $c \colon  A \rightarrow \R$ given by
    \[
    c(a) = \begin{cases}
    \val(a); & a \text{ is basic }\\
    0; & a \text{ is nonbasic}
    \end{cases}
    \]
    is a  linearizing cost function  in reduced form
\end{lemma}

\begin{proof}
    Let $a = (u, v)$ be a strongly basic arc. We claim that $a$ has the property $(\pi)$ iff for any two  $s$-$u$-paths $P$, $Q$  the two-path system $(u,P,Q,N_u,a \cdot N_v)$ is linearizable with respect to $f$. Indeed, note that by \cref{obs:linearizability-two-paths}, the   two-path system above is linearizable with respect to $f$  iff $f(P \cdot a \cdot N_v) + f(Q \cdot N_u) = f(P \cdot N_u) + f(Q \cdot a \cdot N_v)$. The latter equation is equivalent to $\val(a,Q) = \val(a,P)$. Recalling that  the latter equality  
     holds for every pair of $P, Q$ iff $a$ has the property $(\pi)$ completes the proof of the claim.
    
    Now, assume that some strongly basic arc $(u,v)$ does not have the property $(\pi)$. Then, the corresponding two-path system $(u,P,Q,N_u,a \cdot N_v)$ is not linearizable with respect to $f$ and therefore,  $(G,s,t,f)$ is  not linearizable.
    
    Finally, assume that every strongly basic arc has the property $(\pi)$. 
  In the proof of \cref{thm_linearization_characterization}
    we  use the linearizability assumption   only for  specific two-path systems  of the  form $(u, P, Q, N_u, a \cdot N_v)$,  where   $a = (u, v)$ is some strongly basic arc. Thus,  if the property $(\pi)$ holds for all strongly basic arcs, then  each such specific two-path system is linearizable with respect to $f$ and the linearizability of $(G,s,t,f)$ follows. Furthermore, the value $c(a)$ of the linearizing cost function in \cref{lincost:equ}   equals  $\val(a)$ for  any arc $a$ which is either  strongly basic  or incident to $s$, while  $c(a)=0$ for any   nonbasic arc $a$. 
    %This completes the proof of the lemma.
    \qed
\end{proof}

\section{A linear time algorithm for the \textsc{Lin}SPP$_d$}\label{algo:sec}
%%%%%%%%%%%%%%%%%
In this section, we describe an algorithm which solves the linearization problem for SPP$_d$ (\textsc{Lin}SPP$_d$) in  $\bigO(m^d)$ time, i.e., in linear time.
The algorithm uses the relationship between the \textsc{Lin}SPP$_d$ and the  \emph{All-Paths-Equal-Cost Problem (APECP)} which we introduce in  Section~\ref{APEC:ssec}.  
%The algorithm itself is described in \cref{subsection:alg-description}.  Finally, in %\cref{subsection:alg-runtime} we discuss    the running time of the algorithm.
In \cref{subsection:alg} we describe the SPP$_d$ algorithm and discuss
its running time.
 \subsection{The All Paths Equal Cost Problem of Order-$d$ (APECP$_d$)}
\label{APEC:ssec}
%An instance $(G,s,t,q_d)$ of the All Paths  Equal Cost Problem of Order-$d$ (APECP$_d$) %consists of  a $\Pst$-covered acyclic digraph with source vertex  $s$, sink vertex $t$ and an  %order-$d$  cost function $q_d$. The task is  to decide  whether all paths from $s$ to $t$ have %the same cost. 
The All Paths  Equal Cost Problem of Order-$d$ (APECP$_d$) is defined as follows.
\begin{center}
%%%%%%%%%%%%%%%%%
\boxxx{\textbf{Problem:}  ALL PATHS EQUAL COST of Order-$d$ (APECP$_d$)
\\[1.0ex]
\textbf{Instance:} An acyclic $\Pst$-covered directed graph $G=(V,A)$ with a source vertex $s$ and a sink vertex $t$, an  integer $d \geq 1$;
an order-$d$   cost function $q_d\colon  \set{B \subseteq A : |B| \leq d} \to \R$.
\\[1.0ex]
\textbf{Question:} Do all $s$-$t$-paths  have the same cost, i.e.\ is there some $\beta \in \R$ such that $\text{SPP}_d(P,q_d) = \beta$ for every path $P$ in $\Pst$?}
%%%%%%%%%%%%%%%%%
\end{center}

In the following we establish a connection between the \textsc{Lin}SPP$_d$ and the APECP$_{d-1}$ for $d\ge 2$. More precisely, we show in Lemma~\ref{lemma:corresponding-instance}  that    an   instance $(G,s,t,q_d)$ of the  \textsc{Lin}SPP$_d$ with an acyclic  $\Pst$-covered digraph $G=(V,A)$ can  be reduced  to $\bigO(m)$ instances of APECP$_{d-1}$, each of them corresponding to exactly one strongly basic arc with respect to some fixed system of nonbasic arcs (see Definitions~\ref{nonbasic:def} and \ref{stronglybasic:def}).  The  APECP$_{d-1}$ instance corresponding to a strongly basic arc  $a = (u, v)$ is defined as follows.
\begin{definition}
\label{def:corresponding-instance}
  The instance $I^{(a)}$ of the APECP$_{d-1}$ corresponding to the strongly basic arc $a=(u,v)$ takes as input   the digraph $G^{(a)} = (V_u, E_u)$ with  source vertex $s' = s$,  sink vertex $t' = u$, where 
 $V_u$ is the set of vertices in $V$ lying on at least one $s$-$u$-path and $A_u$  is the set of arcs in $A$ lying on at least one $s$-$u$-path.
  The  order-$(d-1)$  cost function   $q_{d-1}^{(a)}\colon \{ B \subseteq A_u\colon  |B| \leq d-1\} \to \rz$ is given by 
    \begin{equation} \label{eq:definition_q_e}
        q_{d-1}^{(a)}(B) := \left(\sum_{\substack{C \subseteq N_u\\ |C| \leq d - |B|}}q_d(B \cup C) \right) - \left( \sum_{\substack{C \subseteq a \cdot N_v\\ |C| \leq d - |B|}}q_d(B \cup C) \right). 
    \end{equation}
\end{definition}

\begin{lemma}
\label{lemma:corresponding-instance}
Let $d \geq 2$ and let  $(G, s,t, q_d)$ be an instance of the \textsc{Lin}SPP$_d$ with a fixed system of nonbasic arcs  $N$. 
The APECP$_{d-1}$ instance $I^{(a)}$ corresponding to some strongly basic arc $a$ is a \textsc{YES}-instance iff the arc $a$ has the property $(\pi)$ with respect to $f\colon \Pst\to\rz$ given by $f(P)=SPP_d(P,q_d)$ for $P\in \Pst$. In this case, $\val(a) = \beta$, where $\beta$ is the common cost of all paths in the APECP$_{d-1}$ instance.
\end{lemma}
\begin{proof}[Sketch]
Let $a = (u,v) \in A$ be a strongly basic arc and let $P$ be some $s$-$u$-path in $G$. Then $P$ is  contained in the graph $G_a = (V_u, A_u)$. 
It can be shown that
\begin{align*}
    &\quad \val(a, P) = f(P \cdot N_u) - f(P \cdot a \cdot N_v) = \sum_{\substack{B \subseteq P\\ |B| \leq d-1}}q_{d-1}^{(a)}(B) = f^{(a)}(P)\, ,
\end{align*}
where $f^{(a)}(P)=SPP_{d-1}(P,q^{(a)}_{d-1})$ for any $s$-$u$-path $P$ in $G$.
% We have that
% \begin{align*}
%     &\quad \val(a, P) = f_q(PN_u) - f_q(PaN_v) \\
%     &= \sum_{\substack{F \subseteq PN_u\\ |F| \leq d}}q(F) - \sum_{\substack{F \subseteq PaN_v\\ |F| \leq d}}q(F)\\
%     &= \sum_{k=0}^d \sum_{\substack{B \subseteq P\\ |B| = k}}\sum_{\substack{C \subseteq N_u\\ |C| \leq d - k}}q(B \cup C) - \sum_{k=0}^d \sum_{\substack{B \subseteq P\\ |B| = k}}\sum_{\substack{C \subseteq aN_v\\ |C| \leq d - k}}q(B \cup C)\\
%     &= \sum_{k=0}^{d-1} \sum_{\substack{B \subseteq P\\ |B| = k}}\sum_{\substack{C \subseteq N_u\\ |C| \leq d - k}}q(B \cup C) - \sum_{k=0}^{d-1} \sum_{\substack{B \subseteq P\\ |B| = k}}\sum_{\substack{C \subseteq aN_v\\ |C| \leq d - k}}q(B \cup C) \quad +  (1 - 1)\sum_{\substack{B \subseteq P\\ |B| = d}}q(B)\\
%      &= \sum_{k=0}^{d-1} \sum_{\substack{B \subseteq P\\ |B| = k}}\left(\sum_{\substack{C \subseteq N_u\\ |C| \leq d - k}}q(B \cup C) - \sum_{\substack{C \subseteq aN_v\\ |C| \leq d - k}}q(B \cup C)
%      \right)\\
%      &= \sum_{\substack{B \subseteq P\\ |B| \leq d-1}}q_a(B) = f_{q_a}(P).
% \end{align*}
We conclude that the value $\val(a, P)$ is independent of $P$, if and only if for every path the quantity $f^{(a)}(P)$ does not depend on $P$. The latter condition is equivalent to  $I^{(a)}$  being a \textsc{YES}-instance. Furthermore, if this is the case, then $\val(a) = f^{(a)}(P)$ for any $s$-$u$-path $P$. 
\qed
\end{proof}
\cref{lemma:property-pi,lemma:corresponding-instance} imply that an instance $(G,s,t,q_d)$ of the SPP$_d$ with an acyclic digraph $G$ is linearizable  iff each  instance $I^{(a)}$ of the APECP$_{d-1}$ corresponding to some strongly basic arc $a$   (with respect to some fixed system of nonbasic arcs)  is a YES-instance.
Thus, an instance of the  \textsc{Lin}SPP$_d$ can be reduced to $\bigO(m)$ instances  of the APECP$_{d-1}$.  
Next, in Lemma~\ref{lemma:reduction-APEC} we show that each instance of the APECP$_{d-1}$  can  be reduced to an instance of  the \textsc{Lin}SPP$_{d-1}$.
First, we define a specific cost function as follows. 
\begin{definition}
Let $G = (V, A)$ be a $\Pst$-covered acyclic digraph and $\beta \in \R$. The function $\const_\beta : A \rightarrow \R$  assigns cost $\beta$ to every arc incident to the source $s$, and $0$ to all other arcs.
\end{definition}

\begin{lemma}
\label{lemma:reduction-APEC}
    Let $G = (V, A)$ be a $\Pst$-covered acyclic digraph with source vertex $s$ and sink vertex $t$ and let  $N \subseteq A$ a fixed system of nonbasic arcs. Let $q_d$ be an order-$d$  cost function. The instance $(G,s,t,q_d)$ of the APECP$_d$ problem is a YES-instance iff the instance $(G,s,t, q_d)$ of SPP$_d$ is linearizable and $\const_\beta$ is its unique linearizing function in reduced form (with respect to $N$).
\end{lemma}
\begin{proof}
    Clearly,  $\const_\beta$ is a linearizing function iff   all paths have the same cost $\beta$. % Further, if all paths have the same cost $\beta$, then  $\const_\beta$ is a linearizing function. 
    Furthermore, observe that all  arcs incident to the source do not belong to $N$. Therefore $\const_\beta$ is in reduced form with respect to $N$. In fact, by \cref{lemma:nonbasic-arcs}  $\const_\beta$ is the unique linearizing  functions in reduced form, and $\red(c') = \const_\beta$ for all other linearizing functions $c'$.
\qed
\end{proof}

%\subsection{Description of the algorithm}
%\label{subsection:alg-description}
%%%
%An important part of the algorithm is solving the \emph{All-Paths-Equal-Cost Problem (APEC)}.
% \begin{corollary}\todo{SL: make this no corollary but just text}
% An acyclic digraph with given arc interaction costs is linearizable if and only if (with respect to some fixed system of nonbasic arcs) for each strongly basic arc its corresponding APEC instance is a YES-instance.
% \end{corollary}

% \begin{proof}
%     This follows immediately from \cref{lemma:property-pi,lemma:corresponding-instance}.
%     \qed
% \end{proof}
%Consider an instance $(G,s,t,q_d)$ of the \textsc{Lin}SPP$_d$ with an acyclic $\Pst$ covered %digraph $G$, with source vertex $s$, sink vertex $t$ and order-$d$ cost function $q_d$. We %first fix some system of nonbasic arcs $N$ and solve the instance $I^{(a)}$ of the  %APECP$_{d-1}$ problem given in \cref{def:corresponding-instance} for each strongly basic arc %$a$ according to \cref{lemma:corresponding-instance}. Each $I^{(a)}$  can again be reduced to 
%an instance of \textsc{Lin}SPP$_{d-1}$ according to \cref{lemma:reduction-APEC}.
%By iterating this process we end up with APECP problems of degree $1$ that can be easily solved %by  dynamic programming. The dynamic program is based on the fact that in a $\Pst$-covered %acyclic digraph with a cost function $f\colon \Pst\ro\rz$ all $s$-$t$-paths have the same cost %if and only if  for every vertex $v$ all $s$-$v$-paths have the same cost. 

\subsection{The linear time \textsc{Lin}SPP$_d$ algorithm}
\label{subsection:alg}
Our \textsc{Lin}SPP$_d$ algorithm ${\cal A}$ works as follows.
Consider an instance $(G,s,t,q_d)$ of the \textsc{Lin}SPP$_d$ with an acyclic $\Pst$-covered digraph $G$, with source vertex $s$, sink vertex $t$ and order-$d$ cost function $q_d$. We first fix some system of nonbasic arcs $N$ and construct the instance $I^{(a)}$ of the  APECP$_{d-1}$ problem given in \cref{def:corresponding-instance} for each strongly basic arc $a$. 
Then,  we  check each instance  $I^{(a)}$ for being a  \textsc{YES}-instance and
do this by reducing  $I^{(a)}$  to  an instance of \textsc{Lin}SPP$_{d-1}$ according to \cref{lemma:reduction-APEC}.
By iterating this process we eventually end up with APECP problems of degree $1$ that can be easily solved by  dynamic programming. The dynamic program is based on the fact that in a $\Pst$-covered acyclic digraph with a cost function $f\colon \Pst \rightarrow \rz$ all $s$-$t$-paths have the same cost iff  for every vertex $v$ all $s$-$v$-paths have the same cost. 
\smallskip

It is not hard to implement the algorithm described above
in $\bigO(d^2m^{d+1})$ time.
With a careful implementation it is possible  to achieve  a better result.
\begin{theorem}
The \textsc{Lin}SPP$_d$ on acyclic digraphs can be solved in $\bigO(d^2m^d)$ time.
\end{theorem}
For the sake of brevity we  refer to the full version of the paper for the proof of the theorem.
Here we just point out the necessity of 
an efficient   computation of  the $I^{(a)}$  instances  for all strongly basic arcs as defined in \cref{def:corresponding-instance}.To this end we   efficiently compute  the values 
\begin{equation*}
    \gamma(B, x) := \sum_{\substack{C \subseteq N_x\\ |C| \leq d - |B|}}q(B \cup C).
\end{equation*}
for all sets $B \subseteq A$ of arcs with $|B| \leq d-1$ and all vertices $x \in V \setminus \set{s}$. These values are then used to efficiently compute the  cost functions $q^{(a)}_{d-1}$ in \cref{eq:definition_q_e}. Further, with  a careful management of the quantities involved  in the computation of the linearizing functions (see \cref{lemma:reduction-APEC}) we obtain a linear time algorithm. 
%\smallskip
%
 Note that  the input size required to encode the cost function $q_d$ equals $\sum_{k=0}^d \binom{m}{k} \geq \nicefrac{m^d}{d!}$. Thus,  
 %running time of 
$\bigO(d^2m^d)$ is linear in the input size and hence  optimal if  $d$ is considered a constant, like for example in the QSPP. 

\section{The subspace of linearizable instances}
\label{sec:subspace}
Let $d\in \nz$, $d\ge 2$, and a $\Pst$-covered acyclic digraph  $G = (V, A)$ with source vertex $s$ and sink vertex $t$  be  fixed. 
Let $H^{(d)} := \set{B \subseteq H \mid |B| \leq d}$ be the set of all subsets of at most $d$ arcs in arc set $H\subseteq A$.
Every order-$d$ cost function $q_d\colon A^{(d)}\to \rz$ can be uniquely represented by a vector $x \in \R^{A^{(d)}}$ with $q_d(F)=x_F$ for all $F\in A^{(d)}$, and vice-versa. Thus, each instance $(G,s,t,q_d)$ can be identified with the corresponding vector $x\in \rz^{A^{(d)}}$ and we will say that $x\in \R^{A^{(d)}}$ is an instance of the SPP$_d$. 
%It is straightforward  to see that if $x, y \in \R^{A^{(d)}}$ are   linearizable instances of the  SPP$_d$,  then $\mu x + \nu y$ is also  a  %linearizable instance  of the SPP$_d$,
%for all $\mu, \nu \in \R$.
It is straightforward to see that the  linearizable instances of the SPP$_d$ on the fixed digraph  $G$ form a linear subspace ${\cal L}_d$ of $\R^{A^{(d)}}$. 

Methods to compute this subspace 
are useful in
B\&B algorithms 
for the SPP$_d$ as they can be applied to compute
better lower bounds along the lines of what Hu and Sotirov~\cite{huSo2021} did for general quadratic binary programs.
Hu and Sotirov showed that for $d = 2$ a basis of ${\cal L}_d$ can be computed in polynomial time \cite[Prop. 5]{huSo2021}.
We extend their result 
to the case  $d > 2$.

\begin{theorem}
Let $G = (V, A)$ be  a $\Pst$-covered, acyclic digraph with source vertex $s$ and sink vertex $t$ and let   $d \in \N$ be a constant.  A basis of the subspace ${\cal L}_d$  of the  linearizable instances of the SPP$_d$  can be computed in polynomial time.  \end{theorem}

\begin{proof}[Sketch]
The proof idea is to specify  a  $k\in \nz$ and  a  matrix $M$ of 
polynomially bounded dimensions, 
such that for  $f: \R^{A^{(d)}} \rightarrow \R^k$ with  $f(x)=Mx$,  we have:
$f(x) = 0$ 
iff
$x$ is  a linearizable instance of the SPP$_d$. Thus,  the linearizable instances $x$ of the SPP$_d$ 
form $\ker(M)$ which can be efficiently computed.

The construction of $M$ is done iteratively and exploits the relationship between $SPP_d$ and $APECP_{d-1}$ similarly as in the algorithm 
${\cal A}$ from \cref{subsection:alg}.
In particular we use   the following two facts:
%%%
\begin{itemize}
\item[(i)] For each strongly basic arc $a = (u,v)$, the function which maps $x \in \R^{A^{(d)}}$ to $q^{(a)}_{d-1}\colon A_u^{(d-1)} \to \R$ 
 is linear (see \cref{eq:definition_q_e} and recall \cref{def:corresponding-instance} for $A_u$).
 \item[(ii)] The function $c \mapsto \red(c)$ (defined after \cref{lemma:nonbasic-arcs}) is linear.
 \end{itemize}
Using (i) and (ii)  iteratively as in   
algorithm $\cal A$, we  show by induction that for each strongly basic arcs $a = (u,v)$ and each $d \geq 2$ there 
exist $k_a\in\nz$ and a linear function $g_a : \R^{A_u^{(d-1)}} \rightarrow \R^{k_a}$ 
s.t.\ $g_a(x) = 0$
iff if $x$ corresponds to a YES-instance of 
APECP$_{d-1}$.
Then we  construct the  linear function $g'_a$ on the same domain as $g_a$, by setting  $g_a'(x) = \beta$ whenever $g_a(x) = 0$, where $\beta$ is the common path cost of
the corresponding instance $x$ of APECP$_{d-1}$.
 Next we show that for  each vertex $u$ there exists a $k_u \in \N$ and a linear function
$f_u \colon \R^{A_u^{(d-1)}}\to\R^{k_u}$ such that $f_u(x) = 0$ iff $x$ is
a linearizable instance of SPP$_{d-1}$ corresponding to APECP$_{d-1}$ (see  \cref{lemma:reduction-APEC}). Then we construct the linear function $f'_u$ on the same domain as $f_u$ by setting   $f'_u(x)$ equal to the  linearizing cost function   of the instance $x$ of the SPP$_{d-1}$ whenever $x$ is linearizable (i.e. when $f_u(x)=0$).
The construction of $M$ is done be repeating these steps iteratively until $d=1$. 
One can ensure that the size of the matrix representations of all involved functions stays polynomial.
\qed
\end{proof}
\noindent {\bf Acknowledgement.} This research has been supported by the Austrian Science Fund (FWF): W1230.

\bibliographystyle{splncs04}
\bibliography{literature.bib}
\end{document}